\documentclass[11pt]{llncs}
\usepackage{amsmath}
\usepackage{amssymb}
\usepackage{epic,gastex}
\usepackage{array}
\usepackage{enumerate}

\DeclareSymbolFont{rsfscript}{OMS}{rsfs}{m}{n}
\DeclareSymbolFontAlphabet{\mathrsfs}{rsfscript}

\begin{document}

%\title{Synchronizing automata with random inputs\thanks{SOME GRANTS}}
\title{Synchronizing automata with random inputs}

\titlerunning{Synchronizing automata with random inputs}

\author{Vladimir V. Gusev}

\authorrunning{V. V. Gusev}

\tocauthor{V. V. Gusev}

\institute{Institute of Mathematics and Computer Science,\\
Ural Federal University, Ekaterinburg, Russia\\
\email{vl.gusev@gmail.com}}

\maketitle

\begin{abstract}
We study the problem of synchronization of automata with random inputs.
We present a series of automata such that the expected number of steps 
until synchronization is exponential in the number of states.
At the same time, we show that the expected number of letters to synchronize any pair 
of the famous \v{C}ern\'{y} automata is at most cubic in the number of states.
\end{abstract}

\section{Introduction}

A \emph{complete deterministic} \emph{finite automaton} $\mathrsfs{A}$, or simply \emph{automaton}, is
a triple $\langle Q,\Sigma,\delta\rangle$, where $Q$ is a finite \emph{set of states},
$\Sigma$ is a finite \emph{input alphabet}, and $\delta: Q\times \Sigma\mapsto Q$ is a totally 
defined \emph{transition function}. Following standard notation, by $\Sigma^*$ we mean the set 
of all finite words over the alphabet $\Sigma$, including the empty word $\varepsilon$. 
The function $\delta$ naturally
extends to the free monoid $\Sigma^{*}$; this extension is still
denoted by $\delta$. Thus, via $\delta$, every word
$w\in\Sigma^*$ acts on the set $Q$. 
%For each $v\in\Sigma^*$ and each $q\in Q$ we write
%$q\dt v$ instead of $\delta(q,v)$ and let $Q\dt v=\{q\dt v\mid q\in Q\}$. 

An automaton $\mathrsfs{A}$ is called
\emph{synchronizing}, if there is a word $w\in\Sigma^*$ %such that %for every $q, q' \in Q$ 
%we have $q \dt w = q' \dt w$,
which brings all states of the automaton $\mathrsfs{A}$ to a particular one,
 i.e. there exists a state $t \in Q$ such that $\delta(s,w)=t$ for every $s \in Q$.
Any such word $w$ is said to be a \emph{reset} (or \emph{synchronizing}) \emph{word}
for the automaton $\mathrsfs{A}$. The minimum length of reset words for $\mathrsfs{A}$
is called the \emph{reset threshold} of $\mathrsfs{A}$.
Note, that the language $\mathrsfs{L}$ of synchronizing words of the automaton $\mathrsfs{A}$ 
is a \emph{two-sided ideal}, i.e. $\Sigma^*\mathrsfs{L}\Sigma^* = \mathrsfs{L}$.
We say that that the word $w$ synchronizes a pair $\{s,t\}$ if $\delta(s,w)=\delta(t,w)$.

%This notion dates back to the sixties.
%Back then automata served as a mathematical models of devices working in discrete time.
%such as computers or relay control systems. This led to the following 
%natural problem:
%how can we restore control over such a device, if we do not know its current state? Reset word
%answers this question, since after applying it the state of the device becomes known.
%Nowadays 

%THEY ARE USEFULL
Synchronizing automata serve as transparent and natural models of error-resistant
systems in many applied areas such as robotics, coding theory, and bioinformatics.
At the same time, synchronizing automata surprisingly arise in some parts of pure mathematics:
algebra, symbolic dynamics, and combinatorics on words. See recent surveys
by Sandberg \cite{Sa05} and Volkov \cite{Vo08}
for a general introduction to the theory of synchronizing automata. %The applications
%of synchronizing automata in coding theory and connections with symbolic dynamics
%are elegantly described in the recent book \cite{BePeReu09}
%by Berstel, Perrin and Reutenauer.
%BIG QUESTION

The interest to the field is heated also by the famous \emph{\v{C}ern\'{y} conjecture}.
%One of the most important and natural questions related to synchronizing automata is the following:
%given $n$, how big can the reset threshold of an automaton with $n$ states be?
In 1964 \v{C}ern\'{y} exhibited a series $\mathrsfs{C}_n$ of automata with $n$ states whose reset threshold equals 
$(n-1)^2$ \cite{Ce64}. Soon after he conjectured, that this series represents the worst possible case,
i.e. the reset threshold of every $n$-state synchronizing automaton is at most $(n-1)^2$.
%This hypothesis has become known as
%the \emph{\v{C}ern\'{y} conjecture}. 
In spite of its simple
formulation and intensive researchers' efforts, the \v{C}ern\'{y}
conjecture remains unresolved for fifty years. 
%Moreover, no upper bound of magnitude $O(n^2)$ for the 
%reset threshold
%of a synchronizing $n$-state automaton is known so far. 
The best known 
upper bound on the reset threshold of a synchronizing
$n$-state automaton is $\frac{n^3-n}6$ by
Pin~\cite{Pi83}. 
%Though the \v{C}ern\'{y} conjecture is open in general,
%it has been confirmed for some special classes of synchronizing automata, see~\cite{AS09,Du98,Ka03,Tr07,Vo09}. For some classes
%a quadratic upper bound is established, see~\cite{BBP,Rys}.
%Carpi

The focus of this paper is on probabilistic aspects of synchronization.
One general question that was actively studied in the literature is the following:
what are synchronizing properties of a \emph{random automaton}?
Skvortsov and Zaks have shown that a random automaton with sufficiently large 
number of letters is synchronizing with high probability~\cite{SZ10}. Later on,
they proved that a random $4$-letter automaton is synchronizing with a positive probability
that is independent of the number of states~\cite{ZS13}. The last step in this direction seems to be done by Berlinkov~\cite{Be13}.
He has shown that a random automaton over a binary alphabet is synchronizing with high probability.
%Skvortsov and Zaks have shown that a uniformly random automaton with $n$ states
%and $\ln(n)$ letters is synchronizing with high probability. sufficiently many
Another direction within this setting is devoted to reset thresholds of random synchronizing automata.
It was shown in~\cite{SZ10} that a random automaton with large number of letters satisfies 
the \v{C}ern\'{y} conjecture with high probability. Furthermore, computational experiments performed in~\cite{ST11,KKS13}
suggest that expected reset threshold of a random synchronizing automaton is sub-linear.

The setting of the present paper is different. 
%Rather
In our considerations 
%instead of random automata 
we investigate how \emph{random input}
acts on a \emph{fixed} automaton.
%In our considerations automaton is fixed, 
%but the \emph{input is random}. 
Assume that several copies of a synchronizing automaton $\mathrsfs{A}$
simultaneously read a common input from a fixed source 
of random letters. Initially these automata may be in different states. What is the expected number steps $E$
until all copies will be in the same state? 
%of letters these copies have to simultaneously read from a source 
%of random letters until they reach the same state?
%More precisely, given several copies of the same synchronizing automaton
%in different states that read the same source of random letters. How long
%long it will take for them to synchronize. Decoder, due to a mistake . Average restoration time.
We can give the following illustration of this approach. Let $\mathrsfs{D}$ be a decoder of a code.
Due to data transmission errors the decoder $\mathrsfs{D}$ may be in a different state compared to
a correct decoder $\mathrsfs{D}_c$. Then the number $E$ computed for decoders $\mathrsfs{D}$ and $\mathrsfs{D}_c$
represents an average number of steps before recovery of the decoder $\mathrsfs{D}$ after an error.

Our setting heavily depends on a model of a random input. In the present paper we restrict ourselves with
a binary alphabet $\Sigma = \{a,b\}$ and the \emph{Bernoulli model}, i.e.
every succeeding letter is drawn independently 
%from the previous choices 
with probability
$p$ for the letter $a$ and probability $q=1-p$ for the letter $b$.
In section~\ref{sec:null} we present a series of $n$-state automata $\mathrsfs{U}_n$ over $\Sigma$
and a pair $S$
such that the expected number of steps to synchronize $S$ is exponential in $n$.
%that
%require exponential number of steps in $n$ to synchronize a fixed pair.
%At the same time, in section~\ref{sec:cerny} we show that any pair of the famous example $\mathrsfs{C}_n$ 
%by \v{C}ern\'{y} requires only cubic number of steps.
At the same time, in section~\ref{sec:cerny} we show that the expected number of steps to synchronize any pair 
of the famous example $\mathrsfs{C}_n$ by \v{C}ern\'{y} is at most cubic in $n$.
These results reveal that despite the fact that synchronization of $\mathrsfs{C}_n$ is hard in the deterministic 
case, %synchronization of $\mathrsfs{C}_n$ is relatively easy in the random setting.
it is relatively easy in the random setting.

\section{Automata $\mathrsfs{U}_n$ with the sink state}
\label{sec:null}
Let $\Sigma$ be a binary alphabet $\{a, b\}$.
Let $\mathrsfs{U}_n$ be the minimal automaton recognizing the language $L_n$,
where $L_n$ is equal to
$\Sigma^* a^{\frac{n+1}{2}} b^{\frac{n-1}{2}} \Sigma^*$ if $n$ is odd, and
to $\Sigma^* a^{\frac{n}{2}} b^{\frac{n}{2}} \Sigma^*$ if $n$
is even. Note, that the automaton $\mathrsfs{U}_n$ is synchronizing, and its 
language of synchronizing words coincides with $L_n$.

First, we will consider the case when $n$ is odd.
Let us define $\mathrsfs{U}_n$ more formally, see fig.~\ref{fig:U7}. The set of states of $\mathrsfs{U}_n$ is equal
to $\{1,2, \ldots n + 1\}$. The transition function $\delta$ of $\mathrsfs{U}_n$
is defined as follows:
\[
\delta (i, a) = 
\begin{cases}
i + 1,&\mbox{if } i < \frac{n+3}{2}\\
i,&\mbox{if } i = \frac{n+3}{2}\\
2,&\mbox{if } \frac{n+3}{2} < i < n + 1\\
n+1,&\mbox{if } i = n+1;\\
\end{cases}
\delta (i, b) = 
\begin{cases}
1, & \mbox{if } i < \frac{n+3}{2}\\
i + 1, & \mbox{if } \frac{n+3}{2} \leq i < n + 1\\
n+1, & \mbox{if } i = n+1.\\
\end{cases}
\]

\begin{figure}[ht]
 \begin{center}
%  \unitlength=2.4pt
  \unitlength=3pt
    \begin{picture}(80,20)(-5,-10)
    \gasset{Nw=6,Nh=6,Nmr=3,loopdiam=6}
    \node(A1)(0,0){$1$}
    \node(A2)(10,0){$2$}
    \node(A3)(20,0){$3$}
    \node(A4)(30,0){$4$}
    \node(A5)(40,0){$5$}
    \node(A6)(50,0){$6$}
    \node(A7)(60,0){$7$}
    \node(A8)(70,0){$8$}
    \drawloop[loopangle=180](A1){$b$}
    \drawloop[loopangle=90](A5){$a$}
    \drawloop[loopangle=0](A8){$a,b$}
    \drawedge(A1,A2){$a$}
    \drawedge(A2,A3){$a$}
    \drawedge(A3,A4){$a$}
    \drawedge(A4,A5){$a$}
    \drawedge(A5,A6){$b$}
    \drawedge(A6,A7){$b$}
    \drawedge(A7,A8){$b$}
    \drawedge[curvedepth=-4,ELside=r,ELpos=30](A2,A1){$b$}
    \drawedge[curvedepth=-8,ELside=r,ELpos=30](A3,A1){$b$}
    \drawedge[curvedepth=-12,ELside=r,ELpos=30](A4,A1){$b$}
    \drawedge[curvedepth=6,ELpos=40](A6,A2){$a$}
    \drawedge[curvedepth=10,ELpos=40](A7,A2){$a$}
    \end{picture}
\end{center}
\caption{Automaton $\mathrsfs{U}_7$}
\label{fig:U7}
\end{figure}

%It is a well known fact that the language of synchronizing words of $\mathrsfs{U}_n$ is equal
%to $\Sigma^* a^{\frac{n+1}{2}} b^{\frac{n-1}{2}} \Sigma^*$, i.e.
%a word $w$ is synchronizing for $\mathrsfs{U}_n$ if and only if $\delta(1,w)=n+1$, see~\cite{GMP}.
Let $\mathcal{B}(p,q)$ be the source of random letters such that each letter is drawn independently
with probability $p$ for the letter $a$ and probability $q=1-p$ for the letter $b$. Let 
$\mathrsfs{A} = \langle Q,\Sigma,\delta\rangle$ be a synchronizing automaton. We consider the
following random process:\\
\texttt{
1. $S:=Q$\\
2. Until $|S| = 1$ do\\
3. \qquad $x \leftarrow \mathcal{B}(p,q)$\\
4. \qquad $S:=\delta(S,x)$\\
}
We start with the set $S$ equal to the state set $Q$. On each step we draw a random letter $x$ from
the source $\mathcal{B}(p,q)$ and apply it to $S$. We stop when $S$ is a singleton.

In general, we are interested in the average number of steps that this process takes for a given 
automaton $\mathrsfs{A}$. In particular, we have the following theorem.
\begin{theorem}
\label{th:Uodd}
Let $n$ be a positive odd integer. The expected number of letters, that are 
drawn from $\mathcal{B}(p,q)$, until $\mathrsfs{U}_n$ is synchronized,
is equal to $\frac{1}{p^{\frac{n+1}{2}}q^{\frac{n-1}{2}}}$.
\end{theorem}
\begin{proof}
%Recall the following useful fact~\cite{Ma12}.
%\begin{lemma}
%Let $\mathrsfs{L}$ be a two-sided ideal regular language.
%If $\mathrsfs{A}$ is the minimal automaton for $\mathrsfs{L}$,
%then $\mathrsfs{A}$ is synchronizing, and the language of 
%reset words of $\mathrsfs{A}$ is equal to $\mathrsfs{L}$.
%\end{lemma}
It is rather easy to see that the word $w$ synchronizes the automaton $\mathrsfs{U}_n$
if and only if $\delta(1,w)=n+1$. Thus, the average number of steps in our random process equals
the average length %needed to bring the state $1$ to the state $n+1$.
of a random walk that brings the state $1$ to the state $n+1$, where the probability 
of the transition labeled by $a$ is $p$, and the probability of the transition labeled by $b$ is $q$.
%Such a random walk is represented by a Markov chain...
It is well-known how to compute the latter quantity\footnote{It is also called the 
mean absorption time of a Markov chain}~\cite[section 6.2]{Pr13}.
%needed automaton $\mathrsfs{U}_n$ is synchronized 
%as soon as the state $1$ reaches the state $n+1$.
For $1\leq i \leq n + 1$ let $\mu_i$ be the expected length of a random walk that brings the state $i$
to the state $n + 1$.
These quantities necessarily satisfy the following system of equations:
$$\begin{cases}
\mu_1 = p \mu_2 + q \mu_1 + 1 & \qquad (1)\\
\mu_i = p \mu_{i+1} + q \mu_1 + 1, \mbox{ if } 1 \leq i \leq \frac{n+1}{2} & \qquad (2)\\
\mu_{\frac{n+3}{2}} = p \mu_{\frac{n+3}{2}} + q \mu_{\frac{n+5}{2}} + 1 & \qquad (3)\\
\mu_i = p \mu_2 + q \mu_{i+1} + 1, \mbox{ if } \frac{n+5}{2} \leq i \leq n - 1 & \qquad (4)\\
\mu_n = p \mu_2 + q \mu_{n+1} + 1 & \qquad (5)\\
\mu_{n+1} = 0 & \qquad (6)
\end{cases}$$
%We are going to obtain $\mu_1$ for arbitrary values $p$ and $q$.
We will solve this system in several steps:\\
1. Let us show that $\mu_i = \mu_1 - \frac{p^{i-1} - 1}{p^i - p^{i - 1}}$ for $2 \leq i \leq \frac{n+3}{2}$.
Equation $(1)$ implies that this statement is true for $i=2$. % $\mu_2 = \mu_1 - \frac{1}{p}$.
%Assume now that $\mu_i = \mu_1 - \frac{p^{i-1} - 1}{p^i - p^{i - 1}}$. Let us
Suppose now that the statement is true for $\mu_i$. Let us show that it is true for $\mu_{i+1}$.
From equation $(2)$ we get %Then after substitution into equation $\mu_i = p \mu_{i+1} + q \mu_1 + 1$ we get
$\mu_1 - \frac{p^{i-1} - 1}{p^i - p^{i - 1}} = p \mu_{i+1} + q \mu_1 + 1$.
%Equally, $p\mu_1 - \frac{p^{i-1} - 1}{p^i - p^{i - 1}} - 1 = p \mu_{i+1}$.
Therefore, $\mu_1 - \frac{p^{i} - 1}{p^{i+1} - p^{i}} = \mu_{i+1}$.\\
We will denote $\frac{p^{\frac{n+1}{2}} - 1}{p^{\frac{n+3}{2}} - p^{\frac{n+1}{2}}}$ as $C$ in order to simplify
notation. Therefore, $\mu_{\frac{n+3}{2}} = \mu_1 - C$.\\
2. Equation $(3)$ immediately implies 
% of the system we immediately get 
$\mu_{\frac{n+5}{2}} = \mu_{\frac{n+3}{2}} - \frac{1}{q}$.
Therefore, we have
$\mu_{\frac{n+5}{2}} = \mu_1 - C - \frac{1}{q}.$\\
3. Now we will show that $\mu_{i} = \mu_1 - \frac{C}{q^{i - \frac{n+5}{2}}} - \frac{1}{q^{i - \frac{n+3}{2}}}$ 
for $\frac{n+5}{2} \leq i \leq n$.
This statement is true for $i = \frac{n+5}{2}$. Let us show that it is true for every succeeding $i \leq n$.
Since $\mu_2 \overset{(1)}{=} \mu_1 - \frac{1}{p}$ we can rewrite equation $(4)$ in the following way:
$\mu_{i} = p \mu_1 + q \mu_{i+1}$. Our assumption states that
$\mu_{i} = \mu_1 - \frac{C}{q^{i - \frac{n+5}{2}}} - \frac{1}{q^{i - \frac{n+3}{2}}}$.
Therefore, $q\mu_1 - \frac{C}{q^{i - \frac{n+5}{2}}} - \frac{1}{q^{i - \frac{n+3}{2}}} = q \mu_{i+1}$.
Finally, $\mu_{i+1} = \mu_1 - \frac{C}{q^{i - \frac{n+3}{2}}} - \frac{1}{q^{i - \frac{n+1}{2}}}$.\\
Note, we have $\mu_n = \mu_1 - \frac{C}{q^\frac{n-5}{2}} - \frac{1}{q^\frac{n-3}{2}}$.\\
4. Equation $(5)$ and $(6)$ imply $\mu_n = p \mu_1$.\\
Therefore, $\mu_1 = \frac{C}{q^\frac{n-3}{2}} + \frac{1}{q^\frac{n-1}{2}} = \frac{qC+1}{q^{\frac{n-1}{2}}} 
= \frac{1}{p^{\frac{n+1}{2}}q^{\frac{n-1}{2}}}$.
\end{proof}
Slightly modifying the argument of the previous theorem we can obtain a similar result,
when $n$ is even.
\begin{theorem}
Let $n$ be a positive even integer. The expected number of letters, that are drawn 
from $\mathcal{B}(p,q)$, until $\mathrsfs{U}_n$ is synchronized,
is equal to $\frac{1}{p^{\frac{n}{2}}q^{\frac{n}{2}}}$.
\end{theorem}

\section{The \v{C}ern\'{y} automata $\mathrsfs{C}_n$}
\label{sec:cerny}
Now we study a classical example introduced by \v{C}ern\'{y} in 1964~\cite{Ce64}.
Recall the definition of the \v{C}ern\'{y} automaton $\mathrsfs{C}_n$, see fig.~\ref{fig:cerny}. 
The state set of $\mathrsfs{C}_n$ is $Q=\{0,1,\dots,n-1\}$,
and the letters $a$ and $b$ act on $Q$ as follows.
\[\delta(i, a) =\begin{cases}
1 &\text{if } i = 0,\\
i &\text{if } i > 0;
\end{cases}\quad
\delta(i, b) =\begin{cases}
i+1 &\text{if } i<n-1,\\
0 &\text{if } i=n-1.
\end{cases}\]
%The automaton $\mathrsfs{C}_n$ for $n = 7$ is shown in Fig.~\ref{fig:cerny}.
%Here and below we adopt the convention that edges bearing multiple labels
%represent bunches of edges sharing tails and heads. In particular, the edge
%$1\xrightarrow{a,b}2$ in Fig.\,\ref{fig:cerny} represents the two parallel
%edges $1\xrightarrow{a}2$ and $1\xrightarrow{b}2$.
\begin{figure}[ht]
\begin{center}
\unitlength .5mm
\begin{picture}(50,80)(-20,-80)
\gasset{Nw=14,Nh=14,Nmr=7,loopdiam=10}
\node(n0)(0,-5){0}
\node(n1)(-30,-20){6} \node(n2)(30,-20){1}
\node(n3)(-35,-50){5} \node(n4)(35,-50){2}
\node(n5)(-17,-75){4} \node(n6)(17,-75){3}
\drawedge[ELdist=2.0](n1,n0){$b$}
\drawedge[ELdist=1.5](n2,n4){$b$}
\drawedge[ELdist=1.7](n0,n2){$a, b$}
\drawedge[ELdist=1.7](n3,n1){$b$}
\drawedge[ELdist=1.7](n4,n6){$b$}
\drawedge[ELdist=1.7](n6,n5){$b$}
\drawedge[ELdist=1.7](n5,n3){$b$}
\drawloop[ELdist=1.5,loopangle=25](n2){$a$}
\drawloop[ELdist=1.5,loopangle=150](n1){$a$}
\drawloop[ELdist=2.4,loopangle=340](n4){$a$}
\drawloop[ELdist=1.5,loopangle=200](n3){$a$}
\drawloop[ELdist=1.5,loopangle=300](n6){$a$}
\drawloop[ELdist=1.5,loopangle=240](n5){$a$}
\end{picture}
\end{center}
\caption{The automaton $\mathrsfs{C}_7$}
\label{fig:cerny}
\end{figure}
The reset threshold of $\mathrsfs{C}_n$ is equal to $(n-1)^2$, see~\cite{AGV13,Gu,Ce64}.

The goal of the present section is to find the expected number of letters,
that are drawn from $\mathcal{B}(p,q)$, until the pair of states $\{1, \frac{n+1}{2}\}$, when $n$ is odd,
and the pair
$\{1, \frac{n+2}{2}\}$, when $n$ is even, is synchronized.
At the same time, we will see that the expectation for these pairs is the largest among other pairs.

Let $\mathrsfs{A} = \langle Q,\Sigma,\delta\rangle$ be an automaton.
The \emph{pair automaton} $\mathcal{P}(\mathrsfs{A})$ 
%of the automaton $\mathrsfs{A}$ 
is defined as follows. The set of states of $\mathcal{P}(\mathrsfs{A})$
is equal to $\{\{s,t\} \, | \, s \neq t\} \cup \{\mathbf{z}\}$.
The transition function $\delta_{\mathcal{P}}$ of $\mathcal{P}(\mathrsfs{A})$ 
for each $x \in \Sigma$, $s,t \in Q$ is defined by the following rules:
\[ \delta_\mathcal{P}(\{s,t\}, x) =
\begin{cases}
\{\delta(s,x),\delta(t,x)\}, \mbox{if } \delta(s,x) \neq \delta(t,x)\\
\mathbf{z}, \mbox{if } \delta(s,x) = \delta(t,x);
\end{cases}
\delta_\mathcal{P}(\mathbf{z}, x) = \mathbf{z}.
\]
Note, that all words $w$ that synchronize a pair $\{s,t\}$ label a path in $\mathcal{P}(\mathrsfs{A})$ 
from $\{s,t\}$ to $\mathbf{z}$. Furthermore, a word $w$ is synchronizing for $\mathrsfs{A}$ if and only if 
$w$ is synchronizing for $\mathcal{P}(\mathrsfs{A})$. The proof of this easy fact can be found for instance in~\cite{Vo08}.

First, let $n$ be a positive odd integer.
In order to prove the main result of this section we will require another representation of the 
pair automaton of $\mathrsfs{C}_n$. We will denote it by $\mathrsfs{P}_n$, see fig.~\ref{fig:C11_pair}.
%The pair Now to pair automaton of $\mathrsfs{C}_n$. It will convenient to represent it in a slightly
%different form. Later we will prove it that it is indeed the same. 
%We will denote it $\mathrsfs{P}_n$.
%It is different depending on the parity of $\mathrsfs{P}_n$.
%Firstly, we consider the case when $n$ is odd.
The state set of $\mathrsfs{P}_n$ is the set of ordered pairs
\[\{(i, \ell) \, | \, 0 \leq i \leq n - 1, \, 1 \leq \ell \leq \frac{n-1}{2}\} \cup \{\mathbf{z}\}.\]
The transition function $\delta$ is defined as follows.\\
$\delta(\mathbf{z}, x) = \mathbf{z}$ for every $x \in \Sigma$,\\
$\delta((i, \ell), b) = ((i+1) \mod n, \ell)$ for every admissible $i$ and $\ell$.\\
$\delta((i,\ell), a) = (i,\ell)$ for every admissible $i$ and $\ell$ with the exception 
of the following cases:\\
$\delta((0,1), a) = \mathbf{z}$,\\
$\delta((0,\ell), a) = (1,\ell - 1)$ if $2 \leq \ell \leq \frac{n-1}{2}$,\\
$\delta((n - \ell,\ell), a) = (n - \ell,\ell + 1)$ if $1 \leq \ell \leq \frac{n-3}{2}$,\\
$\delta((\frac{n+1}{2}, \frac{n-1}{2}), a) = (1, \frac{n-1}{2})$.

\begin{lemma}
Let $n$ be a positive odd integer. 
The automaton $\mathrsfs{P}_n$ is isomorphic to the pair automaton of $\mathrsfs{C}_n$.
\end{lemma}
\begin{proof}
We will construct the desired isomorphism.
The sink state $\mathbf{z}$ of the pair automaton is mapped to the sink state $\mathbf{z}$ of $\mathrsfs{P}_n$.
Let $\{s,t\}$ be an arbitrary pair of states. Let $\delta_{\mathrsfs{C}}$ be the transition
function of the automaton $\mathrsfs{C}_n$. There is a positive integer $m$ that satisfies equations
$\delta_{\mathrsfs{C}}(s, b^{m}) = t$ and $\delta_{\mathrsfs{C}}(t, b^{n - m}) = s$. Let $\ell$
be the minimum of $m$ and $n-m$. Since $n$ is odd $m \neq n-m$.
Let 
$$i = 
\begin{cases} 
s, & \mbox{if } \delta_\mathrsfs{C}(s, b^{\ell}) = t\\
t, & \mbox{if } \delta_\mathrsfs{C}(t, b^{\ell}) = s.
\end{cases}$$
Then the pair $\{s,t\}$ of the pair automaton is mapped to the state
$(i, \ell)$ of the automaton $\mathrsfs{P}_n$. It is easy to check that the presented mapping is an isomorphism.
\end{proof}

\begin{figure}[ht]
 \begin{center}
%  \unitlength=2.4pt
  \unitlength=3pt
%    \begin{picture}(80,120)(0,0)
    \begin{picture}(80,105)(-15,0)
    \gasset{Nw=7,Nh=7,Nmr=3.5}
%    \node(A1)(0,0){$1$}
%    \node(A2)(10,0){$2$}
%    \node(A3)(20,0){$3$}
%    \node(A4)(30,0){$4$}
%    \node(A5)(40,0){$5$}
%    \node(A6)(50,0){$6$}
%    \node(A7)(60,0){$7$}
%    \node(A8)(70,0){$8$}
\node(410)(0,100){10,5}
\node(49)(0,90){9,5}
\node(48)(0,80){8,5}
\node(47)(0,70){7,5}
\node(46)(0,60){6,5}
\node(45)(0,50){5,5}
\node(44)(0,40){4,5}
\node(43)(0,30){3,5}
\node(42)(0,20){2,5}
\node(41)(0,10){1,5}
\node(40)(0,0){0,5}
\node(310)(10,100){10,4}
\node(39)(10,90){9,4}
\node(38)(10,80){8,4}
\node(37)(10,70){7,4}
\node(36)(10,60){6,4}
\node(35)(10,50){5,4}
\node(34)(10,40){4,4}
\node(33)(10,30){3,4}
\node(32)(10,20){2,4}
\node(31)(10,10){1,4}
\node(30)(10,0){0,4}
\node(210)(20,100){10,3}
\node(29)(20,90){9,3}
\node(28)(20,80){8,3}
\node(27)(20,70){7,3}
\node(26)(20,60){6,3}
\node(25)(20,50){5,3}
\node(24)(20,40){4,3}
\node(23)(20,30){3,3}
\node(22)(20,20){2,3}
\node(21)(20,10){1,3}
\node(20)(20,0){0,3}
\node(110)(30,100){10,2}
\node(19)(30,90){9,2}
\node(18)(30,80){8,2}
\node(17)(30,70){7,2}
\node(16)(30,60){6,2}
\node(15)(30,50){5,2}
\node(14)(30,40){4,2}
\node(13)(30,30){3,2}
\node(12)(30,20){2,2}
\node(11)(30,10){1,2}
\node(10)(30,0){0,2}
\node(010)(40,100){10,1}
\node(09)(40,90){9,1}
\node(08)(40,80){8,1}
\node(07)(40,70){7,1}
\node(06)(40,60){6,1}
\node(05)(40,50){5,1}
\node(04)(40,40){4,1}
\node(03)(40,30){3,1}
\node(02)(40,20){2,1}
\node(01)(40,10){1,1}
\node(00)(40,0){0,1}
\node(0)(50,0){$\mathbf{z}$}
\node[Nframe=n](0f)(40,112){}
\node[Nframe=n](1f)(30,112){}
\node[Nframe=n](2f)(20,112){}
\node[Nframe=n](3f)(10,112){}
\node[Nframe=n](4f)(0,112){}
\drawedge[dash={1}0](010,0f){$b$}
\drawedge[dash={1}0](110,1f){$b$}
\drawedge[dash={1}0](210,2f){$b$}
\drawedge[dash={1}0](310,3f){$b$}
\drawedge[dash={1}0](410,4f){$b$}
\drawedge(10,01){$a$}
\drawedge(20,11){$a$}
\drawedge(30,21){$a$}
\drawedge(40,31){$a$}
\drawedge(00,0){$a$}
\drawedge(010,110){$a$}
\drawedge(19,29){$a$}
\drawedge(28,38){$a$}
\drawedge(37,47){$a$}
\drawedge[curvedepth=-12,ELside=r](46,41){$a$}
\drawedge(00,01){$b$}
\drawedge(01,02){$b$}
\drawedge(02,03){$b$}
\drawedge(03,04){$b$}
\drawedge(04,05){$b$}
\drawedge(05,06){$b$}
\drawedge(06,07){$b$}
\drawedge(07,08){$b$}
\drawedge(08,09){$b$}
\drawedge(09,010){$b$}
\drawedge(10,11){$b$}
\drawedge(11,12){$b$}
\drawedge(12,13){$b$}
\drawedge(13,14){$b$}
\drawedge(14,15){$b$}
\drawedge(15,16){$b$}
\drawedge(16,17){$b$}
\drawedge(17,18){$b$}
\drawedge(18,19){$b$}
\drawedge(19,110){$b$}
\drawedge(20,21){$b$}
\drawedge(21,22){$b$}
\drawedge(22,23){$b$}
\drawedge(23,24){$b$}
\drawedge(24,25){$b$}
\drawedge(25,26){$b$}
\drawedge(26,27){$b$}
\drawedge(27,28){$b$}
\drawedge(28,29){$b$}
\drawedge(29,210){$b$}
\drawedge(30,31){$b$}
\drawedge(31,32){$b$}
\drawedge(32,33){$b$}
\drawedge(33,34){$b$}
\drawedge(34,35){$b$}
\drawedge(35,36){$b$}
\drawedge(36,37){$b$}
\drawedge(37,38){$b$}
\drawedge(38,39){$b$}
\drawedge(39,310){$b$}
\drawedge(40,41){$b$}
\drawedge(41,42){$b$}
\drawedge(42,43){$b$}
\drawedge(43,44){$b$}
\drawedge(44,45){$b$}
\drawedge(45,46){$b$}
\drawedge(46,47){$b$}
\drawedge(47,48){$b$}
\drawedge(48,49){$b$}
\drawedge(49,410){$b$}

%    \drawloop[loopangle=180](A1){$b$}
%    \drawloop[loopangle=90](A5){$a$}
%    \drawloop[loopangle=0](A8){$a,b$}
%    \drawedge(A1,A2){$a$}
%    \drawedge(A2,A3){$a$}
%    \drawedge(A3,A4){$a$}
%    \drawedge(A4,A5){$a$}
%    \drawedge(A5,A6){$b$}
%    \drawedge(A6,A7){$b$}
%    \drawedge(A7,A8){$b$}
%    \drawedge[curvedepth=-4,ELside=r](A2,A1){$b$}
%    \drawedge[curvedepth=-8,ELside=r](A3,A1){$b$}
%    \drawedge[curvedepth=-12,ELside=r](A4,A1){$b$}
%    \drawedge[curvedepth=6](A6,A2){$a$}
%    \drawedge[curvedepth=10](A7,A2){$a$}
    \end{picture}
\end{center}
\caption{Pair automaton of $\mathrsfs{C}_{11}$}
\label{fig:C11_pair}
\end{figure}

Now we are ready to formulate the main result of this section.

\begin{theorem}
Let $n$ be a positive odd integer. The expected number of letters, that are drawn 
from $\mathcal{B}(p,q)$, until the pair $\{1, \frac{n+1}{2}\}$ of $\mathrsfs{C}_n$ is synchronized,
is equal to $\frac{(n-1)((n-1)^2 + q(3n-5) + 4q^2)}{8pq^2}$.
\end{theorem}
\begin{proof}
It is not hard to see that a word $w$ labels a path from $(i,\ell)$ to $\mathbf{z}$ in the automaton $\mathrsfs{P}_n$
if and only if the word $w$ synchronizes the pair $\{i, (i + \ell) \mod n\}$ of the automaton $\mathrsfs{C}_n$.
Thus, the expected number of letters until the pair $\{1, \frac{n+1}{2}\}$ is 
synchronized is equal to the expected length of a random walk in automaton $\mathrsfs{P}_n$
from the state $(1,\frac{n-1}{2})$ to the state $\mathbf{z}$, where the probability of the transition labeled
by $a$ is $p$, and the probability of the transition labeled by $b$ is $q$.
%As in section~\ref{sec:null} we compute expectations before absorption.
For $0\leq i \leq n -1$ and $1 \leq \ell \leq \frac{n-1}{2}$ 
let $\mu_{i, \ell}$ be the expected length of a random walk that 
brings the state $(i,\ell)$ of $\mathrsfs{P}_n$ to the state $\mathbf{z}$.
As in the proof of the theorem~\ref{th:Uodd} these values have to satisfy
a particular system of linear equations, see~\cite[section 6.2]{Pr13}. 
For convenience, we will split this system into three parts.
%Then they have to satisfy the following system of equations that we will split into three parts
%for convenience.
The first part:
\[
\left\{
\begin{array}{l l r}
\mu_{0,1} = q \mu_{1,1} + 1 & & (1)\\
\mu_{i,1} = p \mu_{i,1} + q \mu_{i+1, 1} + 1,& \mbox{if } 1 \leq i \leq n-2 & \qquad (2)\\
\mu_{n-1,1} = p \mu_{n-1,2} + q \mu_{0,1} + 1 & & (3)\\
\mu_{\mathbf{z}} = 0 & &
\end{array}
\right.\]
The second part, $2 \leq \ell \leq \frac{n-3}{2}$:
\[
\left\{
\begin{array}{l l r}
\mu_{0,\ell} = p\mu_{1,\ell - 1} + q\mu_{1,\ell} + 1 & & (4)\\
\mu_{i,\ell} = p \mu_{i,\ell} + q \mu_{i+1, \ell} + 1,& \mbox{  if } 1 \leq i \leq n - \ell -1 & \qquad (5)\\
\mu_{n - \ell,\ell} = p \mu_{n - \ell,\ell + 1} + q \mu_{n-\ell +1, \ell} + 1,&  & \qquad (6)\\
\mu_{i,\ell} = p \mu_{i,\ell} + q \mu_{i+1, \ell} + 1,& \mbox{  if } n - \ell + 1 \leq i \leq n-2 & \qquad (7)\\
\mu_{n-1,\ell} = p \mu_{n-1,\ell} + q \mu_{0, \ell} + 1,& & \qquad (8)\\
\end{array}
\right.\]
And the third part:
\[
\left\{
\begin{array}{l l r}
\mu_{0,\frac{n-1}{2}} = p\mu_{1,\frac{n-3}{2}} + q\mu_{1,\frac{n-1}{2}} + 1 & & (9)\\
\mu_{i,\frac{n-1}{2}} = p \mu_{i,\frac{n-1}{2}} + q \mu_{i+1, \frac{n-1}{2}} + 1,& \mbox{  if } 1 \leq i \leq \frac{n-1}{2} & \qquad (10)\\
\mu_{\frac{n+1}{2},\frac{n-1}{2}} = p \mu_{1,\frac{n-1}{2}} + q \mu_{\frac{n+3}{2}, \frac{n-1}{2}} + 1,&  & \qquad (11)\\
\mu_{i,\frac{n-1}{2}} = p \mu_{i,\frac{n-1}{2}} + q \mu_{i+1, \frac{n-1}{2}} + 1,& \mbox{  if } \frac{n+3}{2} \leq i \leq n-2 & \qquad (12)\\
\mu_{n-1,\frac{n-1}{2}} = p \mu_{n-1,\frac{n-1}{2}} + q \mu_{0, \frac{n-1}{2}} + 1,& & \qquad (13)\\
\end{array}
\right.\]

Let us resolve the first part.
Applying equations $(2)$ in successive order we get
$\mu_{1,1} \overset{(2)}{=} \mu_{n-1,1} + \frac{n-2}{q} \overset{(3)}{=} p \mu_{n-1,2} + q \mu_{0,1} + 1 + \frac{n-2}{q}$.
Since $\mu_{n-1,2}\overset{(8)}{=} \mu_{0,2} + \frac{1}{q} \overset{(4)}{=} p\mu_{1,1} + q\mu_{1,2} + 1 + \frac{1}{q}$
and $\mu_{0,1} \overset{(1)}{=} q\mu_{1,1} + 1$
we have
$\mu_{1,1}= p (p\mu_{1,1} + q\mu_{1,2} + 1 + \frac{1}{q}) + q (q\mu_{1,1} + 1) + 1 + \frac{n-2}{q}$.
After trivial simplification, using the fact that $1 - p^2 - q^2 = 2pq$, we obtain 
\[2\mu_{1,1}=\mu_{1,2}+\frac{n-p}{pq^2} \qquad\qquad (14)\]

Let us focus on the second part.
Let $2 \leq \ell \leq \frac{n-3}{2}$.
Applying equations $(5)$ several times in successive order we get
$\mu_{1, \ell} \overset{(5)}{=} \mu_{n-\ell,\ell} + \frac{n-\ell-1}{q} \overset{(6)}{=} 
p \mu_{n-\ell,\ell+1} + q \mu_{n-\ell+1,\ell} + 1 + \frac{n-\ell-1}{q}$.
Since $\mu_{n-\ell,\ell+1} \overset{(7\,\mbox{\footnotesize{or}}\,12)}{=} \mu_{n-1,\ell+1} + \frac{\ell - 1}{q} 
\overset{(8\,\mbox{\footnotesize{or}}\,13)}{=} \mu_{0,\ell+1} + \frac{\ell}{q}
\overset{(4 \,\mbox{\footnotesize{or}}\,9)}{=} p\mu_{1,\ell} + q\mu_{1,\ell + 1} + 1 + \frac{\ell}{q}$ 
and $\mu_{n - \ell + 1, \ell} \overset{(7)}{=} \mu_{n - 1, \ell} + \frac{\ell - 2}{q}
\overset{(8)}{=} \mu_{0,\ell} + \frac{\ell - 1}{q}
\overset{(4)}{=} p\mu_{1,\ell-1} + q\mu_{1,\ell} + 1 + \frac{\ell - 1}{q}$ we have
$\mu_{1,\ell} = p(p\mu_{1,\ell} + q\mu_{1,\ell + 1} + 1 + \frac{\ell}{q}) 
+ q(p\mu_{1,\ell-1} + q\mu_{1,\ell} + 1 + \frac{\ell - 1}{q}) + 1 + \frac{n-\ell-1}{q}$.
After simplification we obtain the following equation:
\[2\mu_{1,\ell}=\mu_{1,\ell+1}+\mu_{1,\ell-1}+\frac{n-p}{pq^2} \qquad\qquad (15)\]

Let us resolve the third part.
Applying equations $(10)$ in successive order we get
$\mu_{1,\frac{n-1}{2}} \overset{(10)}{=} 
\mu_{\frac{n+1}{2},\frac{n-1}{2}} + \frac{n-1}{2q}
\overset{(11)}{=} p\mu_{1,\frac{n-1}{2}} + q\mu_{\frac{n+3}{2},\frac{n-1}{2}} + 1 + \frac{n-1}{2q}$.
Since $\mu_{\frac{n+3}{2},\frac{n-1}{2}} \overset{(12)}{=} 
\mu_{n-1,\frac{n-1}{2}} + \frac{n-5}{2q} \overset{(13)}{=}
\mu_{0,\frac{n-1}{2}} + \frac{n-3}{2q} \overset{(9)}{=}
p\mu_{1,\frac{n-3}{2}} + q\mu_{1,\frac{n-1}{2}} + 1 + \frac{n-3}{2q}$ we have
$\mu_{1,\frac{n-1}{2}} = p\mu_{1,\frac{n-1}{2}} 
+ q(p\mu_{1,\frac{n-3}{2}} + q\mu_{1,\frac{n-1}{2}} + 1 + \frac{n-3}{2q}) + 1 + \frac{n-1}{2q}$.
After an easy simplification we obtain the following equation:
\[\mu_{1,\frac{n-1}{2}} = \mu_{1,\frac{n-3}{2}} +\frac{q^2 + \frac{n-1}{2}q + \frac{n-1}{2}}{pq^2} \qquad (16)\]

Summing up equations $(14)$,$(15)$ for $2 \leq \ell \leq \frac{n-3}{2}$, and $(16)$ we
obtain the following equation:
\[ \mu_{1,1} = \frac{n-3}{2}\cdot\frac{n-p}{pq^2} + \frac{q^2 + \frac{n-1}{2}q + \frac{n-1}{2}}{pq^2} \qquad (17)\]

Now we can show that
\[\mu_{1,\ell} = \ell\mu_{1,1} - \frac{\ell(\ell - 1)}{2}\cdot\frac{n-p}{pq^2} \qquad (18)\]
Equation $(14)$ serves as the induction base. 
Using equation $(15)$ we make the induction step.
%Induction step can be done 
%Induction can be easily seen via direct computation from $(15)$.

From equation $(18)$ for $\ell = \frac{n-1}{2}$ we get
$\mu_{1,\frac{n-1}{2}} = \frac{n-1}{2}\mu_{1,1} - \frac{(n-1)(n-3)}{8}\cdot\frac{n-p}{pq^2}$.
Using $(17)$ after tedious simplification we get the final result:
\[ \mu_{1,\frac{n-1}{2}} = \frac{(n-1)((n-1)^2 + q(3n-5) + 4q^2)}{8pq^2} \qquad (19)\]
\end{proof}

% minimize[{( (n - 1)( (n-1)^2 + (1-p)(3 n - 5) + 4 (1 - p)^2 ) )/(8 p (1-p)^2), 1>=p>=0, n>2},p]
% Root[-7n^4+16n^3-2n^2-16n+(128n-128)#1^3+(-108n^4+216n^3-12n^2-240n+160)#1^2+(54n^4-120n^3 + 12n^2+120n-66)#1+9&,1]
Note, that the leading term of $\mu_{1,\frac{n-1}{2}}$ is equal to $\frac{n^3}{8pq^2}$.
It is easy to see, that the minimum of $\frac{1}{8pq^2}$ is reached at $p=\frac{1}{3}$.
Therefore, the expected
number of letters until the pair $\{1, \frac{n+1}{2}\}$ of $\mathrsfs{C}_n$ is synchronized
is close to the minimum for the source of random letters $\mathcal{B}(\frac{1}{3},\frac{2}{3})$.
%the source of random letters $\mathcal{B}(\frac{1}{3},\frac{2}{3})$ synchronizes the pair $\{1, \frac{n+1}{2}\}$
%of automaton $\mathrsfs{C}_n$ as fast as possible (for arbitrary odd $n$).
In this case we have
\[\mu_{1,\frac{n-1}{2}} = \frac{27n^3}{32} - \frac{27n^2}{32} - \frac{15n}{32} + \frac{15}{32}\]
%You can see that it is cubic in $n$. Leading coefficient is equal to $\frac{1}{8pq^2}$. Minimum
%of this function is equal to $\frac{27}{32}$ when $p=\frac{1}{3}$.
%So the best option will be the following: 
%\[\frac{27n^3}{32} - \frac{27n^2}{32} - \frac{15n}{32} + \frac{15}{32}\]

\begin{theorem}
Let $n$ be a positive even integer. The expected number of letters, that are drawn 
from $\mathcal{B}(p,q)$, until the pair $\{1, \frac{n+2}{2}\}$ of $\mathrsfs{C}_n$ is synchronized,
is equal to $\frac{n( (n-1)(n-2) +q(3n-6) +4q^2)}{8pq^2}$.
\end{theorem}
\begin{proof}
The proof this theorem is similar to a proof of a previous one and we will omit it.
The main difference lies in the equations $(9)-(13)$.
So, instead of $(16)$ we will get
\[\mu_{1,\frac{n}{2}} = \mu_{1,\frac{n-2}{2}} + \frac{\frac{n-2}{2}+q}{pq} \qquad (20)\]
From $(18)$ and $(20)$ we can obtain the following equation:
\[ \mu_{1,\frac{n}{2}} = \frac{n( (n-1)(n-2) +q(3n-6) +4q^2)}{8pq^2}.\]
%And the best option again is $p=\frac{1}{3}$, so
%\[\frac{27n^3}{32} - \frac{27n^2}{32} - \frac{6n}{32}\]
\end{proof}

As before, the expected number of letters until the pair $\{1, \frac{n+2}{2}\}$ of $\mathrsfs{C}_n$ is synchronized 
is close to the minimum for the the source $\mathcal{B}(\frac{1}{3},\frac{2}{3})$:
\[\frac{27n^3}{32} - \frac{27n^2}{32} - \frac{6n}{32}\]
\section{Conclusion}
The expected number of steps until synchronization of the automata $\mathrsfs{U}_n$ 
is exponential in the number of states.
At the same time, the expected number of steps to synchronize any pair 
of the \v{C}ern\'{y} automata $\mathrsfs{C}_n$ is at most cubic in the number of states.
These results reveal that despite the fact that synchronization of $\mathrsfs{C}_n$ is hard in the deterministic 
case, it is relatively easy in the random setting.

%
%
%MASSLENNIKOVA, hard to all.
%Surprisingly, \v{C}ern\'{y} automaton is not extreme from the point of view of random synchronization.
%It can be synchronized quite fast, using just cubical number of steps.
%
%Expected number of steps before absorption may be seen as a measure of complexity of ideal languages.
%Infimum as inherent difficulty.
%It can be studied in a framework of descriptive complexity. Essentially,
%we have found this complexity for a language of synchronizing words of \v{C}ern\'{y} automata.
%We have also presented another example showing that it can be exponential in the number of states.


\begin{thebibliography}{99}
%\bibitem{AS09}
%Almeida, J., Steinberg, B.: Matrix mortality and the \v{C}ern\'{y}--Pin conjecture. Developments in Language Theory, Lect.\ Notes Comput.\
%Sci. 5583, 67--80 (2009)

\bibitem{AGV13}
Ananichev, D.\,S., Gusev, V.\,V., Volkov M.\,V.: Primitive digraphs with large exponents and
slowly synchronizing automata. Journal of Mathematical Sciences (US), 192(3), 263--278
(2013)

%\bibitem{AVZ}
%Ananichev, D.S., Volkov, M.V., Zaks, Yu.I.: Synchronizing automata
%with a letter of deficiency 2. Theor.\ Comput.\ Sci. 376, 30--41 (2007)

%\bibitem{BBP}
%B\'{e}al, M.-P., Berlinkov, M.V., Perrin, D.: A quadratic upper bound on the size of a synchronizing word
%in one-cluster automata. Int. J. Found. Comput. Sci. 22(2), 277--288 (2011)

\bibitem{Be13}
Berlinkov, M. V.: On the probability of being synchronizable. http://arxiv.org/abs/1304.5774 (2013)

\bibitem{Ce64}
\v{C}ern\'{y}, J.: Pozn\'{a}mka k homog\'{e}nnym eksperimentom s
kone\v{c}n\'{y}mi automatami. Matematicko-fyzikalny \v{C}asopis
Slovensk.\ Akad.\ Vied 14(3) 208--216 (1964) (in Slovak)

%\bibitem{Du98}
%Dubuc, L.: Sur les automates circulaires et la conjecture de
%\v{C}ern\'y. RAIRO Inform.\ Th\'eor.\ Appl. 32, 21--34 (1998) (in
%French)

\bibitem{Gu}
Gusev, V.V.: Lower bounds for the length of reset words in eulerian automata.
Int. J. Found. Comput. Sci., 24(2), 251--262 (2013)

%\bibitem{Ka03}
%Kari, J.: Synchronizing finite automata on Eulerian digraphs.
%Theoret.\ Comput.\ Sci. 295, 223--232 (2003)

\bibitem{KKS13}
Kisielewicz, A., Kowalski J., Szyku{\l}a, M.: A Fast Algorithm Finding the Shortest Reset Words.
Lect.\ Notes Comput.\ Sci. 7936, 182--196 (2013)

%\bibitem{Mart}
%Martyugin, P.V.: Lower bounds for the length of the shortest carefully synchronizing
%words for two- and three-letter partial automata. Diskretn. Anal. Issled. Oper., 15(4), 44--56 (2008)
%(in Russian)

%\bibitem{Ma12}
%Maslennikova, M.I.: Reset Complexity of Ideal Languages. Int. Conf. SOFSEM 2012, Proc. Volume II, 
%Institute of Computer Science Academy of Sciences of the Czech Republic,  33--44 (2012)


\bibitem{Pi83}
Pin, J.-E.: On two combinatorial problems arising from automata
theory. Ann.\ Discrete Math. 17, 535--548 (1983)

\bibitem{Pr13}
Privault, N.: Understanding Markov Chains. Springer (2013)

%\bibitem{Rys}
%Rystsov, I.K.: Estimation of the length of reset words for automata with simple idempotents.
%Cybernetics and Systems Analysis 36(3), 339--344 (2000)

\bibitem{Sa05}
Sandberg, S.: Homing and synchronizing sequences. Model-Based Testing of Reactive Systems. Lect.\ Notes Comput.\ Sci. 3472, 5--33 (2005)

\bibitem{SZ10}
Skvortsov, E. S., Zaks, Yu.: Synchronizing random automata. Discr. Math. and Theor. Comp. Sci. 12(4), 95--108 (2010)

\bibitem{ST11}
Skvortsov, E., Tipikin, E.: Experimental study of the shortest reset
word of random automata. CIAA 2011, Lect.\
Notes Comput.\ Sci., 6807, 290--298 (2011)

%\bibitem{Tr07}
%Trahtman, A.N.: The \v{C}ern\'y conjecture for aperiodic automata.
%Discrete Math.\ Theor.\ Comput.\ Sci. 9(2), 3--10 (2007)

%\bibitem{Tr11}
%Trahtman, A.N.: Modifying the upper bound on the length of minimal synchronizing word. Fundamentals of Computation Theory, Lect.\ Notes
%Comput.\ Sci. 6914, 173--180 (2011)

\bibitem{Vo08}
Volkov, M.V.: Synchronizing automata and the \v{C}ern\'{y}
conjecture.
%In: Mart\'\i{}n-Vide, C.; Otto, F.; Fernau, H. (eds.),
Languages and Automata: Theory and Applications, Lect.\ Notes
Comput.\ Sci. 5196, 11--27 (2008)

\bibitem{Vo09}
Volkov, M.V.: Synchronizing automata preserving a chain of partial
orders. Theoret.\ Comput.\ Sci. 410, 2992--2998 (2009)

\bibitem{ZS13}
Zaks, Yu. I., Skvortsov, E. S.: Synchronizing random automata on a 4-letter alphabet. 
Journal of Mathematical Sciences, 192(3), 303--306 (2013)
\end{thebibliography}
\end{document}